\documentclass[12pt]{article}
\newlength{\mytopmargin}
\newlength{\myleftmargin}
\usepackage{amsmath,amsfonts,amssymb,amsthm}

\theoremstyle{plain}
\newtheorem{theorem}{Theorem}

\newtheorem{prop}[theorem]{Proposition}

\theoremstyle{definition}

\theoremstyle{remark}

\numberwithin{equation}{section}

\begin{document}

\begin{center}
{\bfseries \LARGE Co-rank 1 projections and the randomised \\[1mm] Horn problem }\\[2\baselineskip]
{\LARGE Peter J. Forrester and Jiyuan Zhang\footnote{pjforr@unimelb.edu.au; \, jiyuanz@student.unimelb.edu.au}}\\[.5\baselineskip]

{\itshape ARC Centre of Excellence for Mathematical and Statistical Frontiers,\\
School of Mathematics and Statistics, The University of Melbourne, Victoria 3010, Australia.}\\

\end{center}

\begin{abstract}
Let $\hat{\boldsymbol x}$  be a normalised standard complex Gaussian vector, and project an Hermitian matrix $A$ onto the hyperplane
orthogonal to $\hat{\boldsymbol x}$. In a recent paper Faraut [Tunisian J. Math. \textbf{1} (2019), 585--606] has observed that the 
corresponding eigenvalue PDF has an almost identical structure to the eigenvalue PDF for the rank 1 perturbation
$A + b \hat{\boldsymbol x} \hat{\boldsymbol x}^\dagger$, and asks for an explanation. We provide one by way of a common derivation
involving the secular equations and associated Jacobians. This applies too in related setting, for example when $\hat{\boldsymbol x}$
is a real Gaussian and $A$ Hermitian, and also in a multiplicative setting $A U B U^\dagger$ where $A, B$ are fixed unitary matrices
with $B$ a multiplicative rank 1 deviation from unity, and $U$ is a Haar distributed unitary matrix. Specifically, in each case
there is a dual eigenvalue
problem giving rise to  a PDF  of almost identical structure.
\end{abstract}

\section{Introduction}

Let $A$ be an $n \times n$ complex Hermitian matrix with eigenvalues $a_1 > a_2 > \cdots > a_n$.
Let $\hat{\boldsymbol x}$ denote a random $n \times 1$ vector of standard complex Gaussian entries, normalised to have unit length.
The matrix $\Pi := \mathbb I_n -  \hat{\boldsymbol x} \hat{\boldsymbol x}^\dagger$ is then a co-rank 1 projection onto the hyperplane orthogonal to
$\hat{\boldsymbol x}$. Define
\begin{equation}\label{CO}
B = \Pi A \Pi.
\end{equation}
Interpreting a result of Baryshnikov \cite{Ba01}, we know from \cite{FR02b} that the random matrix $B$ has one zero eigenvalue, and non-zero eigenvalues $\{\lambda_j\}_{j=1}^n$ supported on
\begin{equation}\label{C0}
a_1 > \lambda_1 > a_2 > \lambda_2 > \cdots > \lambda_{n-1} > a_n
\end{equation}
with probability density function (PDF)
\begin{equation}\label{C1}
\Gamma(n) {\prod_{1 \le j < k \le n - 1} ( \lambda_j - \lambda_k) \over \prod_{1 \le j < k \le n } ( a_j - a_k)}.
\end{equation}

With $A$ and $\hat{\boldsymbol x}$ as above, consider next the random matrix
\begin{equation}\label{C2b}
C = A + b  \hat{\boldsymbol x} \hat{\boldsymbol x}^\dagger
\end{equation}
with $b > 0$ a parameter. Interpreting a result of Frumkin and Goldberger \cite[Th.~6.1 and 6.7]{FG06}, restated in
\cite[Th.~5.2]{Fa18}, $C$ has eigenvalues, $\{\mu_i\}_{i=1}^n$ say, supported on
\begin{equation}\label{C3}
\mu_1 > a_1 > \mu_2 > \cdots > \mu_n > a_n
\end{equation}
subject to the constraint
\begin{equation}\label{C4}
\sum_{i=1}^n \mu_i = \sum_{i=1}^n a_i + b,
\end{equation}
and with PDF\footnote{The normalisation constant $\Gamma(n)$ is given as ${1 \over n}$ in \cite{FG06} and repeated in \cite{Fa18}.
This is due to a different convention relating to the implementation of the delta function constraint on the Lebesgue measure in
$\mathbb R^n$ which differs by a factor of $n!$; see the discussion in the second last paragraph of \S1 of
\cite{Fa18}.}
\begin{equation}\label{C5}
\Gamma(n) {1 \over b^{n-1}}
{\prod_{1 \le j < k \le n} (\mu_j - \mu_k) \over \prod_{1 \le j < k \le n} (a_j - a_k) }.
\end{equation}
Observing the similarity between (\ref{C0}), (\ref{C1}) and (\ref{C3}), (\ref{C5}), the recent paper of Faraut \cite{Fa18} in the concluding Section ``Remarks", asks for an
explanation. Here we address this question, showing in Sections \ref{S2.1} and \ref{S2.2} how to give a unified derivation of both results.

The viewpoint of (\ref{C2b}) taken in \cite{FR06,Fa18} is that of a special case of the random matrix sum
\begin{equation}\label{3.1}
U A U^\dagger + V B V^\dagger,
\end{equation}
for $U,V \in U(n)$, chosen with Haar measure. The matrices $A$ and $B$ are fixed Hermitian matrices, and the special case being
considered is $B = {\rm diag} \, (b,0,\dots,0)$. One remarks that $U A U^\dagger$ is the adjoint orbit of the matrix $A$, and similarly the
meaning of $V B V^\dagger$. Also, since matrices in $U(n)$ diagonalise complex Hermitian matrices, the sum in (\ref{3.1}) depends only on
the eigenvalues of $A$ and $B$. Due to this, the question of the eigenvalue PDF of (\ref{3.1}) is a randomised version of Horn's problem
\cite{Ho62}. This randomised version appears to have been first studied in 
\cite{RW90}, in the variant and specialisation of (\ref{3.1}) for which $A,B$ are real symmetric matrices and $U,V \in {\rm SO}(3)$ (see also Section \ref{S2.4} below).
Lie algebraic structures including and generalising (\ref{3.1}) can be found in 
\cite{DRW93}. Recent years has seen a surge of interest in this problem; see e.g.~\cite{Zu18,Fa18,CZ18,ZX17,ZJW18,CMZ19}.
The currentness of this activity provides further motivation for extending our study beyond the question posed in \cite{Fa18}.

The result (\ref{C1}) assumes the eigenvalues of $A$ are all distinct. In the case of the random matrix (\ref{C0}) it is known from
\cite{FR02b} how to extend (\ref{C1}) to the case that $A$ has repeated eigenvalues. We will show in Section \ref{S2.3} how to 
use the methods \cite{FR02b} to calculate the eigenvalue PDF of (\ref{C2b}) in this setting. In Section \ref{S2.4} the case of
(\ref{3.1}) with $A,B$ real symmetric and $U,V \in O(n)$ is considered in the case $B$ having rank 1. In the case $n=3$ this relates to the work
\cite{RW90}. The topic of Section \ref{S2.5} is a randomised multiplicative form of Horn's problem, involving unitary matrices.

When both $A$ and $B$ in (\ref{3.1}) have full rank, Zuber \cite{Zu18} has recently given a multiple integral formula for the corresponding
eigenvalue PDF. This is based on a particular integral over the unitary group due to Harish-Chandra \cite{HC57}, and to 
Itzykson and Zuber \cite{IZ80}, to be referred to as the HCIZ integral. In the case of (\ref{CO}) it is known how to use the latter do derive (\ref{C1}). In Section
\ref{S5} we show how to use the result of \cite{Zu18} to reclaim (\ref{C5}). Following \cite{ZX17,Fa18}, we also draw attention to the relevance of this
integral to the computation of the diagonal entries of the random matrix $U_p A U_p^\dagger$, where $U_p$ is the $p \times n$ matrix
formed from the first $p$ rows $(p \le n)$ of $U \in U(n)$, chosen with Haar measure.

We conclude in Section \ref{S6} with some remarks relating to the analogue of (\ref{3.1}) when $A, B$ are real anti-symmetric and
$U$ real orthogonal.

\section{A unified derivation of (\ref{C1}) and (\ref{C5})}\label{S2}
\subsection{Derivation of (\ref{C1})}\label{S2.1}
We begin by recalling the derivation of the PDF (\ref{C1}) due essentially  to Baryshnikov \cite{Ba01}; see also \cite{FR02b} and
\cite[\S 4.2]{forrester10}. A fundamental point is that the distribution of the random matrix $ \hat{\boldsymbol x} \hat{\boldsymbol x}^\dagger$
in the definition of $\Pi$ in (\ref{CO}) is unchanged by
multiplication on the left or on the right by a unitary matrix. This means that the eigenvalue distribution of $\Pi A \Pi$ is the same
as that when $A$ is replaced by the diagonal matrix of its eigenvalues, which we henceforth assume.

Next, with $B$ as specified in (\ref{CO}), the fact that $\Pi$ is a projector can be used to check that $B$ and
$A \Pi$ have the same eigenvalues. This can be seen from a manipulation of the characteristic polynomial, using (\ref{6.2}) below
with $p=q=n$, $C = - \Pi$, $D = A \Pi$. As a consequence
\begin{align}\label{6.1}
\det (\lambda \mathbb I_n - B ) & = \det (\lambda \mathbb I_n - A \Pi) \: = \det (\lambda \mathbb I_n - A (\mathbb I_n -   \hat{\boldsymbol x} \hat{\boldsymbol x}^\dagger ))  \nonumber \\
& = \det (\lambda \mathbb I_n - A) \det (  \mathbb I_n + (\lambda  \mathbb I_n - A)^{-1} A  \hat{\boldsymbol x} \hat{\boldsymbol x}^\dagger)  \nonumber   \\
& =  \det (\lambda \mathbb I_n - A) \Big ( 1 + \hat{\boldsymbol x}^\dagger ( \lambda  \mathbb I_n  - A)^{-1} A \hat{\boldsymbol x}  \Big ).
\end{align}
To obtain the final equality, the well known formula (see e.g.~\cite[Exercises 5.2 q.2]{forrester10})
\begin{equation}\label{6.2}
\det ( \mathbb I_p + C_{p \times q} D_{q \times p}) =  \det ( \mathbb I_q + D_{q \times p} C_{p \times q}  )
\end{equation}
has been used.  The condition for an eigenvalue $\lambda$ of $\Pi A \Pi$ is thus
\begin{align}\label{6.2a}
0 & = 1  +    \hat{\boldsymbol x}^\dagger ( \lambda \mathbb I - A)^{-1} A \hat{ \boldsymbol x} \nonumber \\
&=   1 -  \hat{\boldsymbol x}^\dagger   \hat{ \boldsymbol x} + \lambda  \hat{\boldsymbol x}^\dagger ( \lambda \mathbb I - A)^{-1}   \hat{ \boldsymbol x}  \nonumber \\ 
& =  \lambda  \hat{\boldsymbol x}^\dagger ( \lambda \mathbb I - A)^{-1}   \hat{ \boldsymbol x}. 
\end{align}

We read off from (\ref{6.2a}) that one eigenvalue is always equal to 0, in keeping with $\Pi$ being of co-rank 1, and that the remaining
eigenvalues are the zeros of the random rational function
\begin{equation}\label{R1}
\sum_{p=1}^n {w_p \over \lambda - a_p}, \qquad w_p := | x_p|^2.
\end{equation}
Here the $x_p$ are the components of $\hat{\boldsymbol x}$. The latter being a normalised standard complex Gaussian vector tells us that
$(|x_1|^2, |x_2|^2,\dots, |x_n|^2)$ is uniformly distributed on the simplex $\sum_{p=1}^n | x_p|^2 = 1$, or equivalently that this latter
vector has Dirichlet distribution, as specified by the  PDF
\begin{equation}\label{D}
{\Gamma(s_1 + \cdots + s_n) \over \Gamma(s_1) \cdots \Gamma(s_n)} \prod_{j=1}^n w_j^{s_j - 1}
\end{equation}
with $w_1,\dots, w_n > 0$ and $\sum_{j=1}^n w_j = 1$, in the case $s_j = 1$ ($j=1,\dots,n$).

Denote the zeros of (\ref{R1}) and thus the non-zero eigenvalues of $\Pi A \Pi$ by $\{ \lambda_j \}_{j=1}^{n-1}$.
Consideration of the graph of (\ref{R1}) establishes the interlacing (\ref{C1}). Also, we can make use of the zeros of
(\ref{R1}) to write
\begin{equation}\label{8.0}
\sum_{p=1}^n {w_p \over \lambda - a_p} = {\prod_{l=1}^{n-1} ( \lambda - \lambda_l) \over
\prod_{l=1}^n ( \lambda - a_l)}.
\end{equation}
Computing the residue at $\lambda = a_j$ gives
\begin{equation}\label{8.1}
w_j =  {\prod_{l=1}^{n-1} ( a_j - \lambda_l) \over
\prod_{l=1, l \ne j}^n ( a_j - a_l)}.
\end{equation}
The measure associated with the distribution of $\{ w_p \}_{p=1}^n$ is read off from the special case 
$s_p=1$ of (\ref{D}), and is thus equal to
\begin{equation}\label{8.1a}
\Gamma(n) dw_1 \cdots dw_{n-1}, 
\end{equation}
subject to the constraints $0 < w_j< 1$ ($j=1,\dots,n-1$) and $\sum_{j=1}^{n-1} w_j < 1$.
 We want to change variables to $\{ \lambda_j \}_{j=1}^{n-1}$. It follows from (\ref{8.1}) by computing
 appropriate partial derivatives to form the Jacobian matrix that
\begin{equation}\label{8.2}
\Gamma(n) dw_1 \cdots dw_{n-1} = \Gamma(n) \Big ( \prod_{j=1}^{n-1} w_j \Big )
\bigg | \det \Big [ {1 \over a_j - \lambda_l } \Big ]_{j,l=1}^{n-1} \bigg | \, d\lambda_1 \cdots d\lambda_{n-1}. 
\end{equation} 

The determinant in (\ref{8.2}) is referred to as the Cauchy double alternant, and has the well known
evaluation (see e.g.~\cite[Eq.~(4.14)]{forrester10})
\begin{equation}\label{9.1}
\det \Big [ {1 \over a_j - \lambda_l} \Big ]_{j,l=1}^{n-1} =
{\prod_{1 \le j < k \le n - 1} (a_j - a_k)(\lambda_j - \lambda_k) \over
\prod_{j,k=1}^{n-1} (a_j - \lambda_k) }.
\end{equation}
Substituting (\ref{9.1}) in (\ref{8.2}) gives (\ref{C1}).

\subsection{Derivation of (\ref{C5})}\label{S2.2}
As for (\ref{C0}), the invariance of 
 $\hat{\boldsymbol x} \hat{\boldsymbol x}^\dagger$ under conjugation by unitary matrices implies that the
 eigenvalue problem for $C$ in (\ref{C2b}) is the same as that when $A$ therein is replaced by its diagonal
 matrix of eigenvalues. We assume this form of $A$.
 
 For the characteristic polynomial of $C$ we have
 \begin{align}\label{9.1a}
\det (\lambda \mathbb I_n - C ) & = \det (\lambda \mathbb I_n - A -b \hat{\boldsymbol x} \hat{\boldsymbol x}^\dagger )   \nonumber   \\
& = \det (\lambda \mathbb I_n - A) \det (  \mathbb I_n - b (\lambda  \mathbb I_n - A)^{-1}   \hat{\boldsymbol x} \hat{\boldsymbol x}^\dagger)  \nonumber   \\
& =  \det (\lambda \mathbb I_n - A) \Big ( 1 - b  \hat{\boldsymbol x}^\dagger ( \lambda  \mathbb I_n  - A)^{-1}  \hat{\boldsymbol x}  \Big ),
\end{align}
 where to obtain the final line use has been made of (\ref{6.2}). It follows that the eigenvalues of $C$, $\{ \lambda_j \}_{j=1}^n$ say, are the
 zeros of the random rational function
\begin{equation}\label{9.2} 
1 - b \sum_{l=1}^n { w_l \over \lambda - a_l} = 0, \qquad w_l := |x_l|^2.
\end{equation}
As in (\ref{R1}) the variables $\{ w_j \}_{j=1}^{n-1}$ have distribution (\ref{D}) with parameters $s_l=1$ ($l=1,\dots,n$).

Consideration of the graph of the LHS of (\ref{9.2}) implies, under the assumption $b > 0$, that the interlacing condition (\ref{C3})
holds with $\{ \mu_j \}_{j=1}^n$ relabelled $\{ \lambda_j \}_{j=1}^n$.
Next, analogous to (\ref{8.0}), by regarding (\ref{9.2}) as a partial fraction expansion involving $\{ \lambda_l \}$ we have
\begin{equation}\label{10.0} 
1 - b \sum_{l=1}^n {w_l \over \lambda - a_l} =
{\prod_{l=1}^n ( \lambda - \lambda_l) \over \prod_{l=1}^n ( \lambda - a_l) }.
\end{equation}
Note that equating the coefficient of $1/\lambda$ on both sides of this expression gives the constraint
(\ref{C4}), telling us in particular that only $\{\lambda_l\}_{l=1}^{n-1}$ are independent.

Computing the residue at $\lambda = a_j$ gives
\begin{equation}\label{10.1} 
- b w_j = {\prod_{l=1}^n (a_j - \lambda_l) \over
\prod_{l=1, l \ne j} (a_j - a_l)}.
\end{equation}
Using this with $\lambda_n$ replaced by $b + a_n + \sum_{j=1}^{n-1} (a_j - \lambda_j)$ in keeping with (\ref{C4})
allows us to compute the appropriate partial derivatives to form the Jacobian matrix for the
change of variables from $\{ w_j\}_{j=1}^{n-1}$ to $\{ \lambda_j \}_{j=1}^{n-1}$ and so obtain
\begin{equation}\label{10.2}
\Gamma(n) dw_1 \cdots dw_{n-1} = \Gamma(n) \Big ( \prod_{j=1}^{n-1} w_j \Big )
\bigg | \det \Big [ {1 \over a_j - \lambda_l} - {1 \over a_j - \lambda_n}  \Big ]_{j,l=1}^{n-1} \bigg | \, d\lambda_1 \cdots d\lambda_{n-1}. 
\end{equation} 
Noting that
\begin{equation}\label{10.3}
 \det \Big [ {1 \over a_j - \lambda_l} - {1 \over a_j - \lambda_n}  \Big ]_{j,l=1}^{n-1} = \prod_{j=1}^{n-1} {( \lambda_j - \lambda_n) \over
( a_j - \lambda_n )} 
 \det \Big [ {1 \over a_j - \lambda_l } \Big ]_{j,l=1}^{n-1} 
\end{equation}
and making use of the Cauchy double alternant determinant evaluation (\ref{9.1}), then substituting the result in (\ref{10.2}) gives
(\ref{C5}).   

We can readily
integrate over $\{x_j\}_{j=1}^{n-1}$ for $n=2$ and $n=3$ and so check that the given
normalisation is consistent with our conventions. In the case $n=2$, (\ref{C5}) with the substitution (\ref{10.2}) reads
\begin{equation}\label{F1}
{1 \over b} {2x_1 - (a_1 + a_2 + b) \over a_1 - a_2}
\end{equation}
while (\ref{C3}) and (\ref{C4}) together imply $a_1 + b > x_1 > {\rm max} \, (a_1, a_2 + b)$. There are thus two distinct cases:
$0 < b < a_1 - a_2$ and $b > a_1 - a_2$. Both integrate to give the value unity, in agreement with the normalisation $\Gamma(n)$.
In the case $n=3$ we specialise to the choice $(a_1,a_2,a_3) = (a,0,-a)$.
The constraints  (\ref{C3}) and (\ref{C4})  then imply $0< x_2 < a$ and $a < x_1 < a + b - x_2$. It is efficient to now use
computer algebra to integrate (\ref{F1}) over these regions, with the value unity resulting, and again confirming
the normalisation as stated in  (\ref{C5}).

\section{Generalisations}
\subsection{Degenerate eigenvalues}\label{S2.3}
Suppose the matrix $A$ in (\ref{C0}) and (\ref{C2b}) is of size $N = \sum_{l=1}^n m_l$ where
$m_l$ is the multiplicity of the eigenvalue $\lambda_l$. In the previous sections it was assumed that $m_l = 1$ $(l=1,\dots,n)$.
With $\Pi$ defined as in (\ref{C0}) but now of size $N \times N$, we know from \cite{FR02b} that the matrix (\ref{C0}) has
one zero eigenvalue, $m_l - 1$ eigenvalues equal to $a_l$ $(l=1,\dots,n)$ and eigenvalues $\{\lambda_l\}_{l=1}^{n-1}$
supported on (\ref{C1}) with PDF
\begin{equation}\label{D1}
{\Gamma(m_1 + \cdots + m_n) \over \Gamma(m_1) \cdots \Gamma(m_n) }
{\prod_{1 \le j < k \le n - 1} (\lambda_j - \lambda_k) \over
\prod_{1 \le j < k \le n} (a_j - a_k)^{m_j + m_k - 1}}
\prod_{j=1}^{n-1} \prod_{p=1}^n | \lambda_j - a_p|^{m_p - 1}.
\end{equation}

It is of interest to compute the eigenvalue PDF of (\ref{C2b}) in this setting, and so to extend (\ref{C3})--(\ref{C5}).
As in the case $m_l=1$, a minor modification of the working used in \cite{FR02b} to derive (\ref{D1}) suffices.
This in turn implies that only a minor modification of the working of Section \ref{S2.2} is required.

With $n$ in (\ref{9.1a}) replaced by $N$ this equation again holds true in the setting of
degenerate eigenvalues. This means that (\ref{9.2}) is again valid, but now with
\begin{equation}\label{wx}
w_l = \sum_{s=1}^{m_l} | x_l^{(s)}|^2,
\end{equation}
where $ x_l^{(s)}$ denotes the components of the vector $\hat{\boldsymbol x}$ in the same row as the
eigenvalue $\lambda_l$ (multiplicity $m_l$) in the matrix of eigenvalues. Since
$\hat{\boldsymbol x}$ is a vector of independent standard complex Gaussian entries normalised to have
length unity, the variables (\ref{wx}) have Dirichlet distribution (\ref{D}) with $s_j = m_j$ ($j=1,\dots,n)$.

The working of (\ref{10.0})--(\ref{10.3}) is independent of the precise values of $\{m_j\}$ and so again
applies. However, relative to the case $m_j = 1$ ($j=1,\dots,n$) there is now a contribution to the PDF 
obtained by substituting (\ref{10.1}) in (\ref{D}), giving a final expression very similar to (\ref{D1}).

\begin{theorem}\label{T1}
Let $A$ be a fixed diagonal matrix with diagonal entries $a_l$ $(a_1 > a_2 > \cdots > a_n)$ each repeated $m_l$
times, and define $N = \sum_{l=1}^n m_l$. Let $\hat{\boldsymbol x}$ be a $N \times 1$ vector of independent standard
complex Gaussians and consider the rank 1 perturbed matrix $C = A + b \hat{\boldsymbol x} \hat{\boldsymbol x}^\dagger$.
This matrix has eigenvalues $a_l$ with multiplicity $m_l - 1$, and remaining eigenvalues $\{ \lambda_l \}_{l=1}^n$ say
supported on (\ref{C3}) and subject to the constraint (\ref{C4}) with the eigenvalue PDF
\begin{equation}\label{D1b}
 {\Gamma(m_1 + \cdots + m_n) \over \Gamma(m_n) \cdots \Gamma(m_n) }{1 \over b^{N-1}}
{\prod_{1 \le j < k \le n } (\lambda_j - \lambda_k) \over
\prod_{1 \le j < k \le n} (a_j - a_k)^{m_j + m_k - 1}}
\prod_{j=1}^{n} \prod_{p=1}^n | \lambda_j - a_p|^{m_p - 1}
\end{equation}
(cf.~(\ref{D1}) and its support (\ref{C1})).
\end{theorem}

\subsection{Adjoint orbits involving real orthogonal matrices}\label{S2.4}
Consider the variant of (\ref{CO}) in which $A$ is an $n \times n$ real symmetric matrix with eigenvalues $a_1 > a_2 > \cdots > a_n$
and $\Pi = \mathbb I_n - \hat{\mathbf x}  \hat{\mathbf x}^T$, with $\hat{\mathbf x}$ a random $n \times 1$ vector of standard real Gaussian
entries, normalised to have unit length. We know from \cite[Cor.~1 with $\beta = 1$, $m_i=1$, $(i=1,\dots,n)$]{FR02b} that the eigenvalue PDF of the $n-1$ non-zero eigenvalues
$\{ \lambda_j \}$ is given by (\ref{D1}) with
$m_l = 1/2$, $l=1,\dots, n$, and is thus equal to 
\begin{equation}\label{D1A}
{\Gamma(n/2) \over \pi^{n/2} } 
{\prod_{1 \le j < k \le n-1 } (\lambda_j - \lambda_k) \over
\prod_{j=1}^{n-1} \prod_{p=1}^n | \lambda_j - a_p|^{1/2}},
\end{equation}
with support given by (\ref{C1}).

The analogous variant of (\ref{3.1}) is to choose the matrices $A, B$ as real symmetric, and $U,V \in O(n)$.
Suppose furthermore that $B = {\rm diag} \, (b,0,\dots, 0)$.
An essential point, already used in the derivation of (\ref{D1A}) as given in \cite{FR02b}, is that the joint distribution
of the components an $n \times 1$ vector of standard real Gaussian
entries is given by the Dirichlet distribution (\ref{D}) with $s_j = 1/2$ ($j=1,\dots,n)$.
Consideration of the working needed to derive (\ref{D1b}) then implies that the eigenvalue PDF is given by (\ref{D1b})
specialised to $m_l = 1/2$, $l=1,\dots, n$, and $N = n/2$. Explicitly, the eigenvalue PDF equals
\begin{equation}\label{D1B}
{\Gamma(n/2) \over \pi^{n/2} } {1 \over b^{n/2 - 1}}
{\prod_{1 \le j < k \le n } (\lambda_j - \lambda_k) \over
\prod_{j=1}^{n} \prod_{p=1}^n | \lambda_j - a_p|^{1/2}},
\end{equation}
supported on (\ref{C3}) and subject to the constraint (\ref{C4}).

The first meaningful case of (\ref{D1B}) is when $n=2$. Introducing the variable $s := \lambda_1 - \lambda_2$,
and the constants $s_{\rm max}:= a_1 - a_2 + b$, $s_{\rm min} := a_2 - a_1 + b$, a simple
calculation gives that (\ref{D1A}) can then be reduced to the density for $s$,
\begin{equation}\label{D1C}
{2 \over \pi} {s \over \sqrt{(s^2 - s_{\rm min}^2) (s_{\rm max}^2 - s^2)}}, \quad s_{\rm min} < s < s_{\rm max}.
\end{equation}
This is a special case ($\beta_2 = 0$) of the density given 
in \cite[Eq.~(36)]{Zu18} for the setting under consideration but now
with $B$  full rank, $B = {\rm diag} \, (b, \beta_2)$.

The case $n=3$ and $b=1$ was first considered in \cite{RW90}. Noting the parametrisation
 \cite[Eq.~(3.1)]{RW90}, it appears that the computed density  \cite[Eq.~(4.2)]{RW90} agrees
 with our (\ref{D1B}), except that the numerator is abscent. This would seem to be a misprint,
 as the specialisation $a_2 = a_3 = 0$ given in  \cite[Eq.~(6.3)]{RW90} contains the
 denominator as is consistent with (\ref{D1B}).

\subsection{A multiplicative randomised Horn's problem}\label{S2.5}
Let $U,V \in U(N)$ be chosen with Haar measure, and let $A,B$ be fixed unitary matrices.
Asking for the eigenvalues of the product matrix $UAU^\dagger VB V^\dagger$ is a randomised form of a multiplicative variant of Horn's
problem (for information and references relating to this
multiplicative Horn's problem without randomisation, see \cite[Sec.~12]{Bh01}).
The facts that unitary matrices are diagonalised by conjugation by other unitary matrices, and that the Haar
measure is invariant under multiplication by fixed unitary matrices, tell us that the eigenvalue PDF of $UAU^\dagger VB V^\dagger$
depends only on the eigenvalues of $A$ and $B$. In the case that $B$ is of the form 
${\rm diag} \, (t,1,\dots,1)$ with $|t|=1$, it is possible to adapt workings already in the literature
\cite{FR06} \cite[Exercises 4.2 q.3]{forrester10} to deduce the eigenvalue PDF (cf.~(\ref{D1b}) and its support).

\begin{prop}\label{P2}
Let $A$ be a fixed $N \times N$ diagonal unitary matrix, with diagonal entries $e^{i \theta_l}$
($0  \le \theta_1 < \theta_2 < \cdots < \theta_n < 2 \pi$), each repeated $m_l$ times so that $N = \sum_{l=1}^n m_l$.
Let $t = e^{i \phi}$ and set $B = {\rm diag} \, (t,1,\dots,1)$. The random unitary product matrix $UAU^\dagger VB V^\dagger$, where $U,V \in U(N)$ are chosen
with Haar measure, has eigenvalues $e^{i \theta_l}$ of multiplicity $m_l - 1$ $(l=1,\dots,n)$. The remaining $n$ eigenvalues,
$\{ e^{i \psi_j} \}_{j=1}^n$ say, are supported on
\begin{equation}\label{supp}
\theta_{i-1} < \psi_i < \theta_i \qquad (i=1,\dots,n; \quad \theta_0 := \theta_n {\rm mod} \, 2 \pi)
\end{equation}
and subject to the constraint
\begin{equation}\label{supp1}
\prod_{l=1}^n e^{i \psi_l} = t \prod_{l=1}^n e^{i \theta_l}.
\end{equation}
They have eigenvalue PDF
\begin{multline}\label{supp2}
{\Gamma(m_1 + \cdots + m_n) \over
\Gamma(m_1) \cdots \Gamma(m_n)}
{1 \over |1 - t|^{N - 1}} \\
\times {\prod_{1 \le j < k \le n} | e^{i \psi_k} - e^{i \psi_j} | \over
\prod_{1 \le j < k \le n}  | e^{i \theta_k} - e^{i \theta_j} |^{m_j + m_k - 1}}
\prod_{j=1}^n \prod_{p=1}^n | e^{i \theta_j} - e^{i \psi_p}|^{m_p - 1}.
\end{multline}
\end{prop}

\begin{proof}
The matrix $UAU^\dagger V B V^\dagger$ has the same eigenvalues as $AWBW^\dagger$ where $W = U^\dagger V \in U(N)$ chosen
with Haar measure. Now $AWBW^\dagger = A (\mathbb I_N + (t-1) \hat{\mathbf w}  \hat{\mathbf w}^\dagger)$ where $\mathbf w$ denotes the first column of
$W$. For the characteristic polynomial of the latter, manipulation analogous to that used in the final two equalities of (\ref{6.1}) gives the factorised form
\begin{equation}\label{f1}
\det ( \lambda I - U) \Big ( t - (t-1) \lambda \sum_{j=1}^n {q_j \over \lambda - \lambda_j} \Big ),
\end{equation}
where $q_j = \sum_{s=1}^{m_l} | w_j^{(s)} |^2$ with $w_j^{(s)}$ denoting the components of the vector $\hat{\mathbf w} $ in the same rows as the
eigenvalue $e^{i \theta_j}$ (multiplicity $m_j$).

It follows immediately from (\ref{f1}) that the eigenvalues of $A$, $e^{i \theta_l}$, with multiplicity greater than 1 remain as eigenvalues
of the product matrix, but now with multiplicity $m_l-1$. It follows too that the remaining eigenvalues are given by the zeros of the second factor.
Writing $\lambda = e^{i \psi}$ and recalling $t = e^{i \phi}$ the implied equation can be written
$$
0 = \cot {\phi \over 2} - \sum_{j=1}^n q_j \cot {\psi - \theta_j \over 2}.
$$
Consideration of the graph of the right hand side of this equation implies the interlacing (\ref{supp}). Also, with
$S = AWBW^\dagger$ by taking the determinant we must have $\det S = \det A \det B$ which is the constraint (\ref{supp1}).

Denoting the second factor in (\ref{f1}) by $C_n(\lambda)$, and setting $\tilde{\lambda}_j = e^{i \psi_j}$ we observe that it permits the rational
function form
$$
C_n(\lambda) = {\prod_{j=1}^n (\lambda - \tilde{\lambda}_j) \over
\prod_{j=1}^n (\lambda - \lambda_j)}.
$$
Taking residues allows us to then deduce
\begin{equation}\label{r1}
-(t-1) \lambda_j q_j =
{\prod_{l=1}^n(\lambda_j - \tilde{\lambda}_l) \over
\prod_{l=1, l \ne j}^n(\lambda_j - \lambda_l) }
\qquad (j=1,\dots,n).
\end{equation}

Using the above, we can read off from the working of \cite[Lemma 2]{FR06} that the Jacobian $J$ for the
change of variables from $\{ q_j \}_{j=1,\dots,n-1} \cup \{t \}$ to $\{ \tilde{\lambda}_j \}_{j=1,\dots,n-1} \cup \{t \}$
is given by
\begin{equation}\label{r2}
J = | 1 - t|^{-(n-1)} \prod_{1 \le j < k \le n} \Big | {\tilde{\lambda}_k -  \tilde{\lambda}_j \over  \lambda_k - \lambda_j} \Big |.
\end{equation}
The probability density for $\{q_j\}$ is given by the Dirichlet distribution  (\ref{D}) with $w_j =  q_j$, $s_j = m_j$ ($j=1,\dots,n)$.
Substituting (\ref{r1}) for $q_j$ and using too (\ref{r2}), by changing variables in the corresponding probability
measure, wedged with $dt$, we read off (\ref{supp2}).

\end{proof}

\section{Some applications of the HCIZ integral }\label{S5}
\subsection{Derivation of (\ref{C5}) }
Zuber \cite{Zu18} has initiated a study of the eigenvalues of the random matrix sum (\ref{3.1}) based on a matrix integral due to Harish-Chandra 
\cite{HC57}, and Itzykson and Zuber \cite{IZ80}. Let the eigenvalues of the Hermitian matrix $X$ be denoted $x:=(x_1,\dots, x_n)$ and those of
the Hermitian matrix $Y$ be denoted
$y =:  (y_1,\dots, y_N)$.
This matrix integral, referred to as the HCIZ integral for short, then reads (see e.g.~\cite[Proposition 11.6.1]{forrester10})
\begin{equation}\label{HCIZ}
\int_U \exp ({\rm Tr} \, U^\dagger X U Y) \, [U^\dagger dU] = {\prod_{j=1}^n} \Gamma(j) \,
{\det [ e^{ x_j y_k} ]_{j,k=1}^n \over \Delta_n(x)  \Delta_n(y)}
\end{equation}
where $[U^\dagger dU]$ denotes the normalised Haar measure for $U(n)$, and for an array $u = (u_1,\dots, u_n)$,
$\Delta_n(u) := \prod_{1 \le j < k \le n} (u_k - u_j)$.
 In this section we will show how the PDF (\ref{C5}) can be derived making use of (\ref{HCIZ}). Use of the later to derive (\ref{C1}) can be
 found in \cite{Ol13}.

Let $X$ and $C$ be $n \times n$ Hermitian matrices $X = [x_{jk} ]_{j,k=1}^n$, $C = [c_{jk} ]_{j,k=1}^n$. Let $X$ be random with
distribution having PDF $f(X)$, and define the Fourier-Laplace transform
\begin{equation}\label{5.1}
\hat{f}_X(C) = \mathbb E_X [ e^{- {\rm Tr} \, CX} ].
\end{equation}
From this definition it is immediate that for $X$ and $Y$ independent
\begin{equation}\label{5.1a}
\hat{f}_{X + Y}(C) = \hat{f}_{X} (C)  \hat{f}_Y (C).
\end{equation}
Noting that
\begin{equation}
{\rm Tr} \, CX = \sum_{j=1}^n c_{jj}^{(r)}  x_{jj}^{(r)}  + 2 \sum_{1 \le j < k \le n} \Big (
c_{jk}^{(r)}  x_{jk}^{(r)} + c_{jk}^{(i)}  x_{jk}^{(i)}  \Big ),
\end{equation}
and writing $(d X) = \prod_{1 \le j \le k \le n} d x_{jk}^{(r)}  \prod_{1 \le j < k \le n} d x_{jk}^{(i)} $
(here the superscripts $(r)$ and $(i)$ denote the real and imaginary parts) (\ref{5.1}) can be rewritten
\begin{equation}\label{5.2}
\hat{f}_X(C) = \int f(X) e^{- \sum_{j=1}^n c_{jj}^{(r)} x_{jj}^{(r)}}
e^{- 2 \sum_{j < k}( c_{jk}^{(r)} x_{jk}^{(r)} +  c_{jk}^{(i)} x_{jk}^{(i)}) } (d X).
\end{equation}
It is assumed that $f$ decays fast enough that this integral converges. Making use of the usual multi-dimensional inverse
Fourier transform shows that (\ref{5.2}) can be inverted to give
\begin{equation}\label{5.3}
f(X) = {1 \over 2^n \pi^{n^2}} \int \hat{f}_X(iC) \exp( i {\rm Tr} \, XC) \, (d C).
\end{equation}

Suppose now that $\hat{f}_X(iC) = \hat{f}_X(iUCU^\dagger)$ for all $U \in U(N)$, and thus is a function of the eigenvalues
$c = (c_1,\dots, c_n)$ only, which is to be denoted by writing $\hat{f}_X(iC) = f_X(ic)$. Then (\ref{5.3}) is a function of the
eigenvalues of $X$ only and we write $f(X) = f(x)$, where $x = (x_1,\dots, x_n)$. In this setting it can be shown, by averaging over
$U$ using the HCIZ integral  (\ref{HCIZ}), that
(\ref{5.3}) reduces to (see e.g.~\cite[Eq.~(1.6)]{KR17})
\begin{equation}\label{5.3a}
f(x) = { (\pi i )^{-n(n-1)/2} \over (2 \pi)^n \Delta_n(x)}
\int_{\mathbb R^n} d c_1 \dots dc_n \, \hat{f}_X(ic)
\Delta_n(c) \prod_{j=1}^n e^{i x_j c_j}.
\end{equation}

Let $Z$ denote the random matrix sum (\ref{3.1}). Making use of (\ref{5.1a}) and then the HCIZ integral to evaluate
$\hat{f}_{U A U^\dagger}(C)$ and $\hat{f}_{V B V^\dagger}(C)$ gives that
\begin{equation}\label{5.3b}
\hat{f}_{Z}(C) = {\prod_{j=1}^n (\Gamma(j))^2 \over \Delta_n(-ia)  \Delta_n(-ib) }
{\det [ e^{-ia_j c_k} ]_{j,k=1}^n   \det [ e^{-ib_j c_k} ]_{j,k=1}^n \over (\Delta_n(c))^2}.
\end{equation}
Replace $X$ by $Z$ in (\ref{5.3a}) and substituting (\ref{5.3b}) with $C$ replaced by $i c$ gives us the PDF of $Z$.
However we seek not the PDF of $Z$ itself but rather the eigenvalue PDF. This can be read off from the former by
recalling that associated with the diagonalisation formula $Z = W^\dagger {\rm diag} \, (z_1,\dots, z_n)W $, where
$\{z_j\}$ are the eigenvalues and $W$ is the matrix of the corresponding eigenvectors, is the 
decomposition of measure
(see e.g.~\cite[Eq.~(1.27) with $\beta = 2$]{forrester10})
$$
(dZ) = (\Delta_n(z))^2 (d z) (W^\dagger d W).
$$
Since the PDF for $Z$ is dependent only on the eigenvalues $z$, we can integrate over $W$ using 
(see e.g.~~\cite[Eq.~(1.28) with $\beta = 2$]{forrester10})
$$
\int (W^\dagger d W) = {\pi^{n(n-1)/2}  \over \prod_{j=1}^n \Gamma(j+1)}.
$$
With $f(z)$ now denoting the eigenvalue PDF of (\ref{3.1}), we obtain
\begin{multline}\label{5.3c}
f(z) = {1 \over (2 \pi)^n} {\prod_{j=1}^n \Gamma(j)   \over i^{3n(n-1)/2} }{\Delta_n(z) \over \Delta_n(a) } \\
\times {1 \over n!} \int \det [ e^{-i a_j c_k} ]_{j,k=1}^n
 \det [ e^{-i b_j c_k} ]_{j,k=1}^n \prod_{l=1}^n e^{i z_l c_l} {1 \over {\Delta_n(b) \Delta_n(c)}} (dc).
 \end{multline}
 This is the result of Zuber \cite[Proposition 1]{Zu18}, obtained by following essentially the
 same steps.
 
 Our specific interest is in the case $b_1, b_2, \dots, b_{n-1} \to 0$ and $b_n = b$. In this limit
 \begin{equation}\label{5.4}
 {\det [ e^{-i b_j c_k} ]_{j,k=1}^n \over \Delta_n(b) } \to
 {(-i)^{(n-2)(n-1)/2} \over \prod_{j=1}^{n-1} \Gamma(j)} {1 \over b^{n-1}}
 \det \begin{bmatrix} [c_k^{j-1} ]_{j=1,\dots, n-1 \atop k=1,\dots, n} \\
 [e^{-ib c_k} ]_{k=1,\dots,n} \end{bmatrix},
 \end{equation}
 which follows by taking the limits successively; see also \cite{Fa15}. More explicitly, note that when it comes to taking $b_l \to 0$, the first $l-1$ rows of the determinant can be
 subtracted in appropriate multiples from row $l$ to reduce its leading term to the one proportional to $b_l^{l-1}$ in its Maclaurin
 expansion. The denominator at this stage consists of $1/((b_l b_{l+1} \cdots b_n)^{l-1} \Delta_{N-l+1}(\{ b_j \}_{j=l}^N))$, so
 the limit $b_l \to 0$ can now be taken immediately by operating on only row $l$ of the determinant.
 
 Consider the product of the factor in the integrand of (\ref{5.3c}) $1/\Delta_n(c)$
 times the determinant in (\ref{5.4}). We see upon making
 of Laplace expansion of the latter, then evaluating the cofactors as Vandermonde products that this quantity,
 which is a symmetric function of $\{c_j\}_{j=1}^n$, that this simplifies to read
 \begin{equation}\label{5.4a} 
 {1 \over \Delta_n(c) } \det \begin{bmatrix} [c_k^{j-1} ]_{j=1,\dots, n-1 \atop k=1,\dots, n} \\
 [e^{-iy c_k} ]_{k=1,\dots,n} \end{bmatrix} = (-1)^{n-1} \sum_{p=1}^n {e^{-i b c_p} \over \prod_{l=1 \atop l \ne p }^N (c_l - c_p)}.
  \end{equation}
  For the product of the other factors in the integrand, we can write
   \begin{equation}\label{5.4b} 
 \det [ e^{-i a_j c_k} ]_{j,k=1}^n    \prod_{l=1}^n e^{i z_l c_l} =
 \det [ e^{-i (a_j - z_k) c_k} ]_{j,k=1}^n.
 \end{equation}
 
 Multiplying together (\ref{5.4a}) and (\ref{5.4b}) we see, upon minor manipulation, that the integrand of (\ref{5.3c})
 in the limiting case of interest 
reduces down to 
  \begin{equation}\label{5.4c}  
  (-1)^{n-1} \sum_{p=1}^n e^{- i b c_p}
  \det \begin{bmatrix} \displaystyle{{e^{-i (a_j - z_k) c_k} \over (c_k - c_p)^{q_{k,p}}}} \end{bmatrix}_{j,k=1}^n, \qquad {\rm where} \: \:
 q_{k,p}  := \left \{ \begin{array}{ll} 1, & k \ne p \\
 0, & k = p. \end{array} \right.
  \end{equation} 
  Consider term $p$ in this sum. The dependence on each $c_l$, $(l \ne p)$ occurs soley in column $l$, so for all these
  variables, the integrations can be done column by column. For these we require
  $$
  {\rm PV} \int_{-\infty}^\infty {e^{-i (a_j - z_k) c} \over c - c_p} \, dc =   -\pi i e^{-i (a_j - z_k) c_p} {\rm sgn} \, (a_j - z_k),
  $$
 which follows by a residue computation. Hence, after simple manipulation of the determinant, and with the integration of each
 $c_p$ in the summand still remaining, we are left with
 $$
 (\pi i )^{n-1} \sum_{p=1}^n e^{- i c_p (b + \sum_{j=1}^n a_j - \sum_{j=1}^n z_j)}
 \det \Big [ \Big (  {\rm sgn} \, (a_j - z_k) \Big )^{q_{j,k}}  \Big ]_{j,k=1}^n.
 $$
 Integrating over $c_p$ is now immediate, showing that the above expression reduces to
   \begin{equation}\label{5.4d}  
   2 \pi^ni^{n-1} \delta \Big ( b + \sum_{j=1}^n a_j - \sum_{j=1}^n z_j \Big ) \sum_{p=1}^n  \det \Big [ \Big (  {\rm sgn} \, (a_j - z_k) \Big )^{q_{j,k}}  \Big ]_{j,k=1}^n.
 \end{equation} 
 Note that the delta function constraint is the requirement (\ref{C4}), with $\{ \mu_i \}$ relabelled $\{z_i\}$.
 
 Consider the determinant in (\ref{5.4d}). We can check that with the $z_j$'s ordered $z_1 > z_2 > \cdots > z_n$, if two of the $a_j$'s
 say $a_q$ and $a_{q'}$ should fall between two consecutive $z_j$'s, or outside of $z_1$ or $z_n$, then rows $q$ and $q'$ are
 identical, so the determinant vanishes. Considering too the requirement of the delta function, we must therefore have the ordering
   \begin{equation}\label{5.4e}  
  z_1 > a_1 > z_2  >  \cdots > z_n > a_n
  \end{equation} 
  which with $\{ \mu_i \}$ relabelled $\{z_i\}$ is (\ref{C3}). With this ordering we can check that only the $p=1$ term is non-zero, with the
  determinant therein equal to
  $$
  \det \begin{bmatrix} 1 &  1 &  1& \cdots & 1 \\
  1 & -1 &  1 & \cdots & 1 \\
  1 & -1 & -1 & \cdots & 1 \\
  \vdots & & & & \vdots \\
   1 & -1 & -1 & \cdots & -1 \end{bmatrix}
   = \det  \begin{bmatrix} 2 & 0 & 0 & \cdots & 0 \\
  2 & -2 & 0 & \cdots & 0 \\
  2 & -2 & -2 & \cdots & 0 \\
  \vdots & & & & \vdots \\
   1 & -2 & -2 & \cdots & -1 \end{bmatrix} = (-2)^{n-1}.
   $$
   Hence (\ref{5.4d}), which is the multiple integral in (\ref{5.3c}) in the limiting case of interest, has the evaluation
   $$
   (2 \pi )^n (-i)^{n-1}n! \delta \Big ( b + \sum_{j=1}^n a_j - \sum_{j=1}^n z_j \Big ),
   $$
   supported on (\ref{5.4e}).
   Substituting this in (\ref{5.3c}), together with the appropriate factors from (\ref{5.4a}), reclaims (\ref{C5}).
   
   An outstanding question along these lines is to develop a method based on matrix transforms
   to similarly reclaim Proposition \ref{P2}; see the works \cite{KK16a,KK16b} for recent results on
   transforms of random product matrices.
   
   \subsection{Distribution of the diagonal entries for $U_p A U_p^\dagger$}
   It is observed in \cite{Fa18, ZX17}, and in fact much earlier in \cite[Eqns.~(3)--(5)]{FK99}
   (see also \cite{Fa05, Fo11})
    that the HCIZ integral (\ref{HCIZ}) has the interpretation as
   the  Fourier-Laplace transform of the distribution of the diagonal entries of the random matrix
   $U B U^\dagger$. Choosing $A = {\rm diag} \, (a_1,\dots, a_p,0,\dots,0)$ corresponds to the 
    Fourier-Laplace transform of the distribution of the diagonal entries of the random matrix
   $U_p B U_p^\dagger$ where $U_p$ is the $p \times n$ matrix formed by the first $p$ rows of
   $U$. Such distributions first appeared in a more general context in the work of Heckman
   \cite{He82}, and are termed Heckman measures.
   
   Let us consider first the case $p=1$. The matrix $U_1 B U_1^\dagger$ is then a scalar quantity,
   corresponding to a particular random quadratic form.
   
   \begin{prop}
   Let $\mathbf z$ be a row vector chosen uniformly at random from the unit sphere in $\mathbb C^n$, and let
   $B$ be an Hermitian matrix with eigenvalues $\{ b_i \}_{i=1}^n$, ordered $b_1 < \cdots < b_n$. 
   Let $h_n(x;b) := (b - x)^{n-2} {\rm sgn} \, (b - x)$. The PDF for the
   distribution of the random
   quadratic form $\mathbf z B \mathbf z^\dagger$
   is supported on $(b_1, b_n)$  and is given by
   \begin{equation}\label{hq}
   {n - 1 \over 2} {1 \over \Delta_n(b)} \det
   \begin{bmatrix} 1 & 1 & \cdots & 1 \\
   b_1 & b_2 & \cdots & b_n \\
   \vdots & \vdots & \cdots & \vdots \\
   b_1^{n-2} & b_2^{n-2} & \cdots & b_n^{n-2} \\
   h_n(b_1,x) &  h_n(b_2,x) & \cdots & h_n(b_n,x)
   \end{bmatrix}.
   \end{equation}
   \end{prop}
   
   \begin{proof}
   Any one row or column of a Haar distributed member of $U(n)$ is uniformly distributed on the complex unit
   sphere in $\mathbb C^n$; see e.g.~\cite{DF17} and references therein. Hence with $U_1$ defined as in the
   text above the statement of the proposition, $U_1 B U_1^\dagger \mathop{=}\limits^{\rm d} \mathbf z B \mathbf z^\dagger$.
   Furthermore, with $A = {\rm diag} \, (ia,0,\dots,0)$ we see that ${\rm Tr} \, A U B U^\dagger = i a U_1 B U_1^\dagger$,
   so in the limit $a_1,a_2,\dots, a_{n-1} \to 0$ with $a_n = i a$ the LHS of the HCIZ integral (\ref{HCIZ}) can be written
   \begin{equation}\label{L1}
   \int_{|| \mathbf z || = 1}  e^{ i a \mathbf z B \mathbf z^\dagger} \, (d \mathbf z).
   \end{equation}
   Taking the limit on the RHS gives
   \begin{equation}\label{L2}
   {n - 1 \over (ia)^{n-1}} {1 \over \Delta_n(b)} \det
   \begin{bmatrix} 
   1 & 1 & \cdots & 1 \\
   b_1 & b_2 & \cdots & b_n \\
   \vdots & \vdots & \cdots & \vdots \\
   b_1^{n-2} & b_2^{n-2} & \cdots & b_n^{n-2} \\
   e^{i a b_1} & e^{i a b_2} & \cdots & e^{i a b_n}
   \end{bmatrix}.
  \end{equation}
  
  The PDF is obtained by multiplying (\ref{L2}) by ${1 \over 2 \pi} e^{-i a  x}$ and integrating over $a$. For the latter task,
  we observe that the only dependence on $a$ in the determinant is in the final row, so we can effectively integrate this
  row. However, the integrals must then be considered as  generalised functions  due to the singularity at the
  origin (alternatively the factor $1/a^{n-1}$ can be replaced by $1/(a + i \delta)^{n-1}$, and the limit $\delta \to 0^+$ be
  taken at the end). Adapting the former viewpoint (this was done is a similar context in the recent work \cite{ZX17}),
  and thus using the generalised integral
  \begin{equation}
  {1 \over 2 \pi} \int_{-\infty}^\infty {e^{i a (b_j - x)} \over a^{n-1}} \, da =
   {i^{n-1} \over 2} {(b_j - x)^{n-2} \over \Gamma(n-1)} {\rm sgn} \, (b_j - x)
  \end{equation}
   gives (\ref{hq}). 
   
  For $x$ outside the interval $(b_1, b_n)$ the $h_n(x;b)$ in the last row simplify to
  $h_n(x;b) = (b - x)^{n-2}$ (after possibly removing an overall sign from the row).
  The determinant can then be seen to vanish, so the support is restricted to $(b_1, b_n)$ in
  keeping with the definition of the quadratic form.
    \end{proof}  
    
    The PDF (\ref{hq}) is a piecewise polynomial of degree $n-2$ in $x$. Such a simple structure is to be
    contrasted with the PDF of the random quadratic form  $\mathbf x B \mathbf x^\dagger$, where
    $\mathbf x$ is a real random vector sampled uniformly at random from the sphere in $\mathbb R^n$
    \cite{PC98,KJ18}, which is a far more complicated function of $x$.
    
    From the original work \cite{He82} the PDF for the distribution of the diagonal entries of $U B U^\dagger$ is
    known as a particular $\binom{n}{2}$-fold convolution. For small $n$ more explicit calculations are also
    possible. For example, with $n=3$, taking the inverse transform of the HCIZ integral we find the PDF
 \begin{multline}
 {12 \over \Delta_3(\beta)} \delta \Big ( \sum_{i=1}^3 (b_i - x_i) \Big ) \bigg (
 (b_2 - b_3) \chi_{b_2 < x_3 < x_2 < x_1 < b_1}      \\
+ (x_3 - b_3)  \chi_{b_2 < x_2 < x_1 < b_1}  \chi_{b_3 < x_3  < b_3} +
 (b_1 - x_1)  \chi_{b_2 <  x_1 < b_1}  \chi_{b_3 < x_3 < x_2 < b_2} \\
 +  (b_1 - b_2) \chi_{b_3 < x_3 < x_2 < x_1 < b_2} \bigg ), 
 \end{multline}
  where we have ordered $  b_3 < b_2 < b_1$ and similarly $ x_3 < x_2 < x_1$. Here $\delta (u)$ denotes the Dirac
  delta function as in (\ref{5.4d}), while $\chi_A = 1$ if $A$ is true, and zero otherwise.
  
  There is a well studied Gaussian version of the above diagonal entries problem, which in the case of
  complex entries has attracted attention for its application to wireless communications
  \cite{HPBMP06,MPEW11}. Thus let $\Sigma$ be a positive
  definite $p \times p$ matrix and $G_{p \times n}$ be a standard Gaussian matrix. The $p \times p$
  matrix $X = \Sigma^{1/2} G G^T \Sigma^{1/2}$ is termed a correlated Wishart matrix. It is straightforward to show
  that the Fourier-Laplace transform of the  distribution of the  diagonal entries of $X$ is equal to
  $$
  \det ( \mathbb I_p - i \Sigma A)^{-\beta n /2}
  $$
  where $\beta = 1$ (2) in the case of real (complex) entries, and $A = {\rm diag} \, (a_1,a_2,\dots, a_p)$.
  For general $\Sigma$, $p,n$  there is no known structured formulae for the inverse transform,
  except in the case $p=2$ when the distribution can be expressed in terms of a Bessel function.
  The reference \cite{HPBMP06} gives this formula in the context of a study of the complexities faced
  in analysing the case $p=3$.

 \section{Concluding remarks}\label{S6}
 
 Unitary matrices diagonalise complex Hermitian matrices, while real orthogonal matrices diagonalise real symmetric matrices.
According to the more general Lie algebraic view of \cite{DRW93} it is natural to consider the random matrix sum (\ref{3.1}) in other
circumstances which share an analogous relation between $A$ and $U$, $B$ and $V$. For example, suppose $A$ (and also $B$) is a real anti-symmetric
matrix. It is well known (see e.g.~\cite{Hu63}) that for $n$ even ($n=2N$ say) there exists an element of $O(2N)$ such that conjugation by
this matrix puts $A$ into the block diagonal form
\begin{equation}\label{3.1a}
{\rm diag} \, \bigg ( \begin{bmatrix} 0 & a_1 \\ - a_1 & 0 \end{bmatrix}, \dots,  \begin{bmatrix} 0 & a_N \\ - a_N & 0 \end{bmatrix} \bigg ).
\end{equation}
The same holds true for $n$ odd ($n=2N+1$ say) with the block diagonal form now reading
\begin{equation}\label{3.1b}
{\rm diag} \, \bigg ( \begin{bmatrix} 0 & a_1 \\ - a_1 & 0 \end{bmatrix}, \dots,  \begin{bmatrix} 0 & a_N \\ - a_N & 0 \end{bmatrix}, [0] \bigg ).
\end{equation}
An integral formula for the eigenvalue PDF of (\ref{3.1}), which makes use of Harish-Chandra's
\cite{HC57} extension of (\ref{HCIZ}) to such settings, has been given in \cite{Zu18}, and the $N=2$ case has been made explicit.

It is known from the works \cite{De10,FILZ17,KFI17} how to compute an explicit eigenvalue PDF for the randomised sum
$A + G^T B G$, where the pair $(A,B)$ are of the form (\ref{3.1a}), or (\ref{3.1b}), with the further requirement that
$N$ is rank 2 (the smallest rank compatible with the structures), and $G$ is a standard real Gaussian matrix. 
Replacing $G$ by a Haar distributed real orthogonal matrix,
we have
found that specialisation of the integral formulas in  \cite{Zu18} to this low rank perturbation setting does not lead to
a simple structured formula analogous to (\ref{C5}) or (\ref{D1B}). Rather the fact that the perturbation is rank 2 leads
to much more complicated structures involving a vast number of terms, conveying little information as to the salient analytic
features (the support, singularities near the boundary etc.).

\section*{Acknowledgements}
This work is part of a research program supported by the Australian Research Council (ARC) through the ARC Centre of Excellence for Mathematical and Statistical frontiers (ACEMS). PJF also acknowledges partial support from ARC grant DP170102028, and JZ acknowledges the support of a Melbourne postgraduate award, and an ACEMS top up scholarship.
We thank J.~Faraut for helpful remarks.

%Some small $n$ cases of the form of the 
%eigenvalue PDF  in these settings are given in \cite{
%
%The standard $2 \times 2$ matrix form of a quaternion is
%\begin{equation}\label{3.1c}
%\begin{bmatrix} z & w \\
%- \bar{w} & \bar{z} \end{bmatrix}, \qquad z,w \in \mathbb C.
%\end{equation}
%Suppose now that $A$ (and also $B$) is a $2N \times 2N$ matrix consisting of such block, and  furthermore
%is anti-Hermitian (i.e.~$A^\dagger = - A$). This class of matrices has eigenvalues in complex conjugate pairs $\pm a_j$
%($a_j > 0$, $j=1,\dots,n$), and are diagonalised by unitary matrices with block structure (\ref{3.1c}). The latter are upon
%similarity transformation by permutation matrices elements of the classical group $USp(2N)$. In Section 3.2 we 
%compute the eigenvalue PDF of (\ref{3.1}) in this setting, assuming furthermore that the eigenvalues of $B$ consist of one
%nonzero pair only.

\providecommand{\bysame}{\leavevmode\hbox to3em{\hrulefill}\thinspace}
\providecommand{\MR}{\relax\ifhmode\unskip\space\fi MR }
\providecommand{\MRhref}[2]{%
  \href{http://www.ams.org/mathscinet-getitem?mr=#1}{#2}
}
\providecommand{\href}[2]{#2}

\end{document}